\documentclass[runningheads]{llncs}
\usepackage{graphicx}

\usepackage{amsmath,amssymb}
\usepackage{xcolor}

\usepackage{booktabs}
\usepackage{caption}
\usepackage{hyperref}
\usepackage{enumerate}
\usepackage{enumitem}
\usepackage{lineno}
\usepackage[vlined]{algorithm2e}

\newcommand{\REMOVE}[1]{}

\newcommand{\eps}{\varepsilon}

\newcommand{\crn}{\text{cr}}

\renewcommand{\phi}{\varphi}

\newcommand{\emphh}[1]{\textbf{#1}}

\newcommand{\rephrase}[3]{\noindent\textbf{#1 #2}.~\emph{#3}}

\newtheorem{observation}{Observation}

\begin{document}
%
%
\title{Crossing Minimization in Perturbed Drawings\thanks{Research supported in part by the NSF awards CCF-1422311 and CCF-1423615,
and the
Austrian Science Fund (FWF): M2281-N35.}}
%
%
\author{Radoslav Fulek\inst{1}
\and
Csaba D. T\'oth\inst{2,3}
}
\authorrunning{R.~Fulek and Cs.D.~T\'oth}
%
\institute{{Institute of Science and Technology, Klosterneuburg, Austria}\\
\email{radoslav.fulek@gmail.com}
\and
California State University Northridge, Los Angeles, CA, USA
\and
Tufts University, Medford, MA, USA\\
\email{cdtoth@eecs.tufts.edu}}
\maketitle              
\begin{abstract}
Due to data compression or low resolution, nearby vertices and edges of a graph drawing may be bundled to a common node or arc.
We model such a ``compromised'' drawing by a piecewise linear map $\varphi:G\rightarrow \mathbb{R}^2$. We wish to perturb $\varphi$ by an arbitrarily small $\eps>0$ into a proper drawing (in which the vertices are distinct points, any two edges intersect in finitely many points, and no three edges have a common interior point) that minimizes the number of crossings. An $\eps$-perturbation, for every $\eps>0$, is given by a piecewise linear map $\psi_\eps:G\rightarrow \mathbb{R}^2$ with $\|\varphi-\psi_\eps\|<\eps$, where $\|.\|$ is the uniform norm (i.e., $\sup$ norm).

We present a polynomial-time solution for this optimization problem when $G$ is a cycle and the map $\varphi$ has no \emphh{spurs} (i.e., no two adjacent edges are mapped to overlapping arcs). We also show that the problem becomes NP-complete (i) when $G$ is an arbitrary graph and $\varphi$ has no spurs, and (ii) when $\varphi$ may have spurs and $G$ is a cycle or a union of disjoint paths.
\keywords{map approximation  \and c-planarity \and crossing number}
\end{abstract}

\section{Introduction}
\label{sec:intro}

A graph $G=(V,E)$ is a 1-dimensional simplicial complex. A continuous piecewise linear map $\varphi:G\rightarrow \mathbb{R}^2$
maps the vertices in $V$ into points in the plane, and the edges in $E$ to piecewise linear arcs between the corresponding vertices. However, several vertices may be mapped to the same point, and two edges may be mapped to overlapping arcs. This scenario arises in applications in cartography, clustering, and visualization, due to data compression, graph semantics, or low resolution.
Previous research focused on determining whether such a map $\varphi$ can be ``perturbed'' into an embedding. Specifically, a continuous piecewise linear map $\phi:G\rightarrow M$ is a \emphh{weak embedding} if, for every $\eps>0$, there is an embedding $\psi_\eps:G\rightarrow M$ with $\|\varphi-\psi_\eps\|<\eps$, where $\|.\|$ is the uniform norm (i.e., $\sup$ norm).
Recently, Fulek and Kyn\v{c}l~\cite{FK17+_ht} gave a polynomial-time algorithm for recognizing weak embeddings,
and the running time was subsequently improved to $O(n \log n)$ for simplicial maps by Akitaya et al.~\cite{AFT18}.
Note, however, that only planar graphs admit embeddings and weak embeddings. In this paper, we extend the concept of $\eps$-perturbations to nonplanar graphs, and seek a perturbation with the minimum number of crossings.

A continuous map $\varphi:G\rightarrow M$ of a graph $G$ to a 2-manifold $M$ is a \emphh{drawing} if (i) the vertices in $V$ are mapped to distinct points in $M$, (ii) each edge is mapped to a Jordan arc between two vertices without passing through any other vertex, and (iii) any two edges intersect in finitely many points.
A \emphh{crossing} between two edges, $e_1,e_2\in E$, is defined as an intersection point between the relative interiors of the arcs $\varphi(e_1)$ and $\varphi(e_2)$. For a piecewise linear map $\varphi:G\rightarrow \mathbb{R}^2$,
let $\crn(\varphi)$ be the minimum nonnegative integer $k$ such that for every $\eps>0$, there exists a drawing $\psi_\eps:G\rightarrow \mathbb{R}^2$ with $\|\varphi-\psi_\eps\|<\eps$ and $k$ crossings, see
Fig.~\ref{fig:simpleExample} for an illustration.

\begin{figure}
\centering
\includegraphics[width=0.8\textwidth]{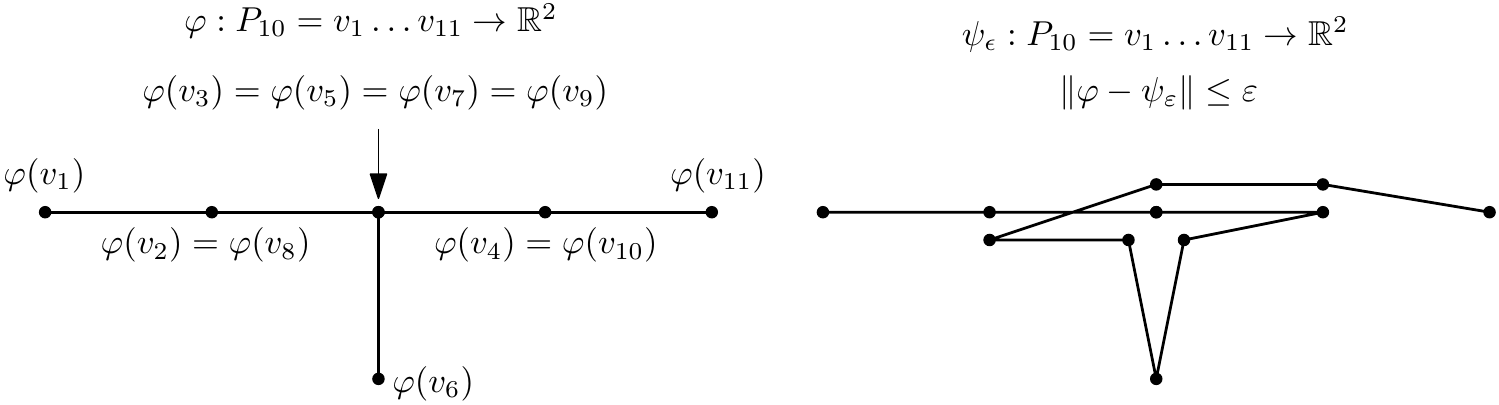}
\caption{An example for a map $\varphi:G\rightarrow \mathbb{R}^2$, where $G=P_{10}$, i.e., a path of length 10, with $\crn(\varphi)=1$ (left); and a perturbation $\psi_\eps$ witnessing that $\crn(\varphi)\le 1$ (right).}
\label{fig:simpleExample}
\end{figure}

It is clear that $\varphi$ is a weak embedding if and only if $\crn(\varphi)=0$. Note also that if $e_1,e_2\in E$ and the arcs  $\varphi(e_1)$ and $\varphi(e_2)$ cross transversely at some point $p\in \mathbb{R}^2$, then $\psi_\eps(e_1)$ and $\psi_\eps(e_2)$ also cross in the $\eps$-neighborhood of $p$ for any sufficiently small $\eps>0$. An $\eps$-perturbation may, however, remove tangencies and partial overlaps between edges.

The problem of determining $\crn(\varphi)$ for a given map $\varphi:G\rightarrow \mathbb{R}^2$ is NP-complete: In the special case that $\varphi(G)$ is a single point, $\crn(\varphi)$ equals the crossing number of $G$, and it is NP-complete to find the crossing number of a given graph~\cite{GJ82} (even if $G$ is a planar graph plus one edge~\cite{CM13}).

In this paper, we focus on the special case that $G$ is a cycle.
A series of recent papers~\cite{AAET17,ChEX15,CDPP09} show that weak embeddings can be recognized in $O(n\log n)$ time. Chang et al.~\cite{ChEX15} identified two features of a map $\varphi:G\rightarrow \mathbb{R}^2$ that are difficult to handle: A \emphh{spur} is a vertex whose incident edges are mapped to the same arc or overlapping arcs, and a \emphh{fork} is a vertex mapped to the relative interior of the image of some nonincident edge (a vertex may be both a fork and a spur). We prove the following results.

\begin{theorem}\label{thm:cycle}
Given a cycle $G=(V,E)$ and a piecewise linear map $\varphi:G\rightarrow \mathbb{R}^2$, where $G$ has $n$ vertices and the image $\varphi(G)$ is a plane graph with $m$ vertices, then $\crn(\varphi)$ can be computed
\begin{enumerate}
\item in $O((m+n)\log (m+n))$ time if $\varphi$ has neither spurs nor forks,
\item in $O((mn)\log (mn))$ time if $\varphi$ has no spurs.
\end{enumerate}
\end{theorem}

As noted above, the problem of determining $\crn(\varphi)$ is NP-complete when $G$ is an arbitrary graph (even if $\varphi$ is a constant map).
We show that the problem remains NP-complete if $G$ is a cycle and we drop the condition that $\varphi$ has no spurs.

\begin{theorem}\label{thm:hardness}
Given $k\in \mathbb{N}$ and a piecewise linear map $\varphi:G\rightarrow \mathbb{R}^2$,
it is NP-complete to decide whether $\crn(\varphi)\leq k$ if  $\varphi:G\rightarrow \mathbb{R}^2$ may have spurs and
\begin{enumerate}
\item $G$ is a cycle, or
\item $G$ is a union of disjoint paths.
\end{enumerate}
\end{theorem}

\medskip\noindent\textbf{Related previous work.}
Finding efficient algorithms for the recognition of weak embeddings
$\varphi:G\rightarrow M$,
where $G$ is an arbitrary graph, was posed as an open problem in~\cite{AAET17,ChEX15,CDPP09}. The first polynomial-time solution for the general version follows from a recent variant~\cite{FK17+_ht} of the Hanani-Tutte theorem~\cite{Ha34_uber,Tutte70_toward}, which was conjectured by M.~Skopenkov~\cite{Sko03_approximability} in 2003 and in a slightly weaker form already by Repov\v{s} and A.~Skopenkov~\cite{ReSk98_deleted} in 1998. Weak embeddings of graphs also generalize various graph visualization models such as \emphh{strip planarity}~\cite{ADDF17} and \emphh{level planarity}~\cite{JLM98_level}; and can be seen as a special case~\cite{AL16_pipes} of the notoriously difficult \emphh{cluster-planarity} (for short, \emphh{c-planarity})~\cite{FCE95a_how,FCE95b_planarity}, whose tractability remains elusive today.

\medskip\noindent\textbf{Organization.} We start in Sec.~\ref{sec:prelim} with preliminary observations that show that determining $\crn(\varphi)$ is a purely combinatorial problem, which can be formulated without metric inequalities.
We describe and analyse a recognition algorithm, proving Theorem~\ref{thm:cycle} in Sec.~\ref{sec:cycles}. We prove NP-hardness by a reduction from 3SAT in Sec.~\ref{sec:hardness}, and conclude in Sec.~\ref{sec:con}.
Omitted proofs are available in the Appendix.

\section{Preliminaries}
\label{sec:prelim}

We rely on techniques introduced in~\cite{AAET17,ChEX15,CDibPP05_cycles,FK17+_ht}, and complement them with additional tools to keep track of edge crossings. A piecewise linear function $\varphi:G\rightarrow \mathbb{R}^2$ is a composition $\varphi=\gamma\circ \lambda$, where $\lambda:G\rightarrow H$ is a continuous map from $G$ to a graph $H$ (i.e., a 1-dimensional simplicial complex) and $\gamma:H\rightarrow \mathbb{R}^2$ is a drawing of $H$.
We may further assume, by subdividing the edges of $G$ if necessary, that the map $\lambda:G\rightarrow H$ is \emphh{simplicial}, that is, it maps vertices to vertices and edges to edges; and $\gamma:H\rightarrow \mathbb{R}^2$ is a straight-line drawing of $H$, where each edge in $E(H)$ is mapped to a line segment. To distinguish the graphs $G$ and $H$ in our terminology, $G$ has \emphh{vertices} $V(G)$ and \emphh{edges} $E(G)$, and $H$ has \emphh{clusters} $V(H)$ and \emphh{pipes} $E(H)$.

A perturbation $\psi_\eps$ of $\varphi$ lies in the $\eps$-neighborhood of $\varphi(G)$. We define suitable neighborhoods for the graph $H$, and the image $\gamma(H)=\varphi(G)$. For the graph $H$ and its drawing $\gamma:H\rightarrow \mathbb{R}^2$, we define the  \emphh{neighborhood} $\mathcal{N}\subset \mathbb{R}^2$ as the union of regions $N_u$ and $N_{uv}$ for every $u\in V(H)$ and $uv\in E(G)$, respectively, as follows. Let $\eps_0>0$ be a sufficiently small constant specified below. For every $u\in V(H)$, let $N_u$ be the closed disk of radius $\eps_0$ centered at $\gamma(u)$.
For every edge $uv\in E(H)$, let $N_{uv}$ be the set of points at distance at most $\eps_0^2$ from $\gamma(uv)$ that lie in the interior of neither $N_u$ nor $N_v$. Let $\eps_0>0$ be so small that for every triple $\{u,v,w\}\subset V(H)$, the disk $N_u$ is disjoint from both $N_v$ and $N_{vw}$, and the regions $N_{uv}$ and $N_{uw}$ are disjoint from each other. (Note, however, that regions ${N}_{uv}$ and ${N}_{u'v'}$ may intersect if the line segments $\gamma(uv)$ and $\gamma(u'v')$ cross.)

Such $\eps_0>0$ exists due to piecewise linearity of $\varphi$ and by compactness. (Indeed, consider the intersection $B_{u,v}$ and $B_{u,w}$ of the boundary of $N_u$ with that of $N_{uv}$ and $N_{uw}$, respectively. Taking $\eps_0$ sufficiently small, we assume that $N_u \cap \gamma(uv)$ and $N_u \cap \gamma(uw)$ are line segments meeting in $u$ at some angle $\alpha\le \pi$. We require $\eps_0<\frac{1}{\pi}\alpha$ since we need $\eps_0^2<\frac{1}{\pi}\eps_0\alpha$ for $B_{u,v}$ and $B_{u,w}$ to be disjoint, and hence $N_{uv}$ and $N_{uw}$.)
By definition, an $\eps$-perturbation of $\varphi=\gamma\circ \lambda$ lies in the neighborhood $\mathcal{N}$ for all $\eps\in (0,\eps_0^2)$.

For the graph $H$ and its drawing $\gamma:H\rightarrow \mathbb{R}^2$, we also define the \emphh{thickening} $\mathcal{H}$, $H\subset \mathcal{H}$,  as a 2-dimensional manifold with boundary as follows. For every $u\in V(H)$, create a topological disk $D_u$, and for every edge $uv\in E(H)$, create a rectangle $R_{uv}$. For every $D_u$ and $R_{uv}$, fix an arbitrary orientation of $\partial D_u$ and $\partial R_{uv}$, respectively. Partition the boundary of $\partial D_u$ into $\deg(u)$ arcs, and label them by $A_{u,v}$, for all $uv\in E(H)$, in the cyclic order around $\partial D_u$ determined by the rotation of $u$ in the the drawing $\gamma(G)$. The manifold $\mathcal{H}$ is obtained by identifying two opposite sides of every rectangle $R_{uv}$ with $A_{u,v}$ and $A_{v,u}$ via an orientation preserving homeomorphism. Note that there is a natural map $\Gamma:\mathcal{H}\rightarrow \mathcal{N}$ such that $\Gamma|_H=\gamma$; $\Gamma$ is a homeomorphism between $D_u$ and $N_u$ for every $u\in V(H)$; and $\Gamma$ maps $R_{uv}$ to $N_{uv}$ for every $uv\in E(H)$.

We reformulate a problem instance $\varphi:G\rightarrow \mathbb{R}^2$ as two functions
$\lambda:G\rightarrow H$ and $\gamma:H\rightarrow \mathbb{R}^2$,
where $G$ and $H$ are abstract graphs,
$\lambda$ is a simplicial map and $\gamma$ is a straight-line drawing of $H$.
A \emphh{perturbation} of the map $\varphi=\gamma\circ \lambda$ is a drawing $\psi=\Gamma\circ \Lambda$, where $\Lambda:G\rightarrow \mathcal{H}$ is a drawing of $G$ on $\mathcal{H}$ with the following properties:
\begin{itemize}
\item[{\rm (P1)}] for every vertex $a\in V(G)$, $\Lambda(a)\in D_{\lambda(a)}$,
\item[{\rm (P2)}] for every edge $ab\in E(G)$, $\Lambda(ab)\subset D_{\lambda(a)}\cup R_{\lambda(a)\lambda(b)}\cup D_{\lambda(b)}$
    such that it crosses the boundary of the disks $D_{\lambda(a)}$ and $D_{\lambda(b)}$ precisely once, and
\item[{\rm (P3)}] all crossing between arcs $\Lambda(e)$, $e\in E(G)$, lie in the disks $D_u$, $u\in V(H)$;
\end{itemize}
and $\Gamma:\mathcal{H}\rightarrow \mathbb{R}^2$ maps the disk $D_u$ injectively into $\mathcal{N}_u$ for all $u\in V(H)$,
and rectangle $R_{uv}$ into ${N}_{uv}$ for all $uv\in E(H)$ (however the rectangles $R_{uv}$ and $R_{u'v'}$ may be mapped to crossing neighborhoods ${N}_{uv}$ and ${N}_{u'v'}$ for two independent edges $uv,u'v'\in E(H)$).

\medskip\noindent\textbf{Combinatorial Representation.}
Properties (P1)--(P3) allow for a combinatorial representation of the drawing $\Lambda:G\rightarrow \mathcal{H}$:
For every pipe $uv\in E(H)$, let $\pi_{uv}$ be a total order of the edges in $\lambda^{-1}[uv]\subseteq E(G)$
in $R_{\lambda(a)\lambda(b)}$; and let $\pi_\Lambda=\{\pi_{uv}:uv\in E(H)\}$ the collection of these total orders.
In fact, we can assume that $\Lambda(G)$ consists of straight-line segments in every rectangle $R_{uv}$,
and every disk $D_u$. The number of crossings in each disk $D_u$ is determined by the cyclic order of the
segment endpoints along $\partial D_u$. Thus the number of crossings in all disk $D_u$, $u\in V(H)$ is
determined by $\pi_\Lambda$.

\medskip\noindent\textbf{Two Types of Crossings.}
The reformulation of the problem allows us to distinguish two types of crossings in a piecewise-linear map
$\varphi:G\rightarrow \mathbb{R}^2$: edge-crossings in the neighborhoods $N_u$, $u\in V(H)$,
and crossings between edges mapped to two pipes that cross each other.

The number of crossings between the edges of $G$ inside a disk $N_u$, $u\in V(H)$,
is the same as the number of crossings in $D_u$, since $\Gamma$ is injective on $D_u$.
We denote the total number of such crossings by

$$\crn_1(\lambda)=\min_{\Lambda} \left(\sum_{u\in V(H)} \text{CR}_\Lambda(u)\right),$$
where $\text{CR}_\Lambda(u)$ is the number of crossings of the drawing $\Lambda(G)$ in the disk $D_u$.

Let the \emphh{weight} of a pipe $e\in E(H)$ be the number of edges of $G$ mapped to $e$, that is, $w(e):=|\lambda^{-1}[e]|$. If the arcs $\gamma(e_1)$ and $\gamma(e_2)$ cross in the plane, for some $e_1,e_2\in E(H)$, then every edge in $\lambda^{-1}[e_1]$ crosses all edges in $\lambda^{-1}[e_2]$. The total number of crossings between the edges of $G$ attributed to the crossings between pipes is

$$\crn_2(\gamma,\lambda)=\sum_{\{e_1,e_2\}\in C} w(e_1) w(e_2),$$
where $C$ is the multiset of pipe pairs $\{e_1,e_2\}$ such that $\gamma(e_1)$ and $\gamma(e_2)$ cross.
It is now clear that

\begin{equation}\label{eq:cr1-2}
\crn(\gamma\circ\lambda)=\crn_1(\lambda)+\crn_2(\gamma,\lambda).
\end{equation}

The operations in Section~\ref{sec:cycles} successively modify an instance $\varphi=\gamma\circ\lambda$
until $H$ becomes a cycle. In this case, it is easy to determine $\crn_2(\gamma,\lambda)$, which is a consequence of the following  folklore lemma.
\begin{lemma}\label{lem:cycle}\cite[Lemma~1.12]{Hass85_intersections}
If $G=C_n$ and $H=C_k$ and $\lambda:G\rightarrow H$ is a simplicial map without spurs, where the cycle $G$ winds around the cycle $H$ precisely $n/k$ times,
then $\crn_1(\lambda)=\frac{n}{k}-1$.
\end{lemma}

\section{Cycles without Spurs}
\label{sec:cycles}

Let $G=C_n$ be a cycle with $n$ vertices, and $H$ an arbitrary abstract graph, $\lambda:G\rightarrow H$ a simplicial map that does not map any two consecutive edges of $G$ to the same edge in $H$, and $\gamma:H\rightarrow \mathbb{R}^2$ a straight-line drawing. In this section, we prove that $\crn(\gamma\circ\lambda)$ is invariant under the so-called \textsf{ClusterExpansion} and \textsf{PipeExpansion} operations.
(Similar operations for weak embeddings have been introduced in~\cite{AAET17,ChEX15,CDibPP05_cycles,FK17+_ht}.)
We show that a sequence of $O(n)$ operations produces an instance in which $H$ is a cycle, where we can easily determine both $\crn_1(\lambda)$ and $\crn_2(\gamma,\lambda)$, hence $\crn(\gamma\circ\lambda)$. \bigskip

\begin{figure}[htbp]
	\centering
	\includegraphics[width=0.8\textwidth]{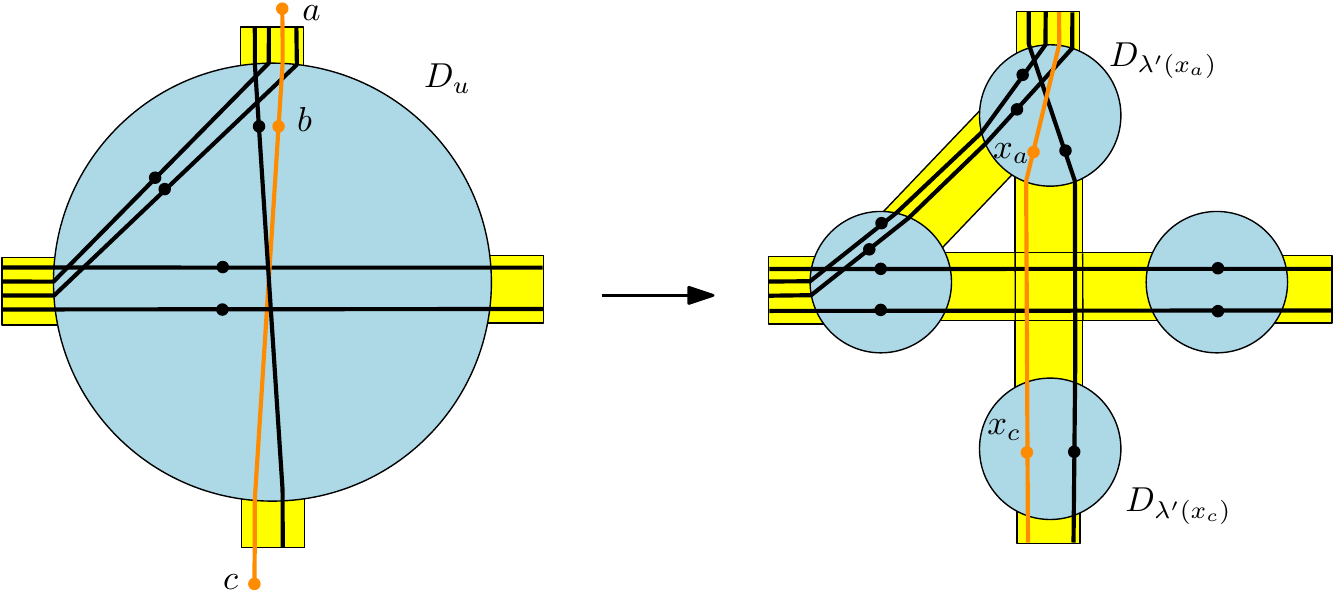}
	\caption{\textsf{ClusterExpansion}$(u)$.}
	\label{fig:expansion}
\end{figure}

\begin{quote}\textsf{ClusterExpansion}$(u)$.
See Figure~\ref{fig:expansion} for an illustration.
(1) Let $D_u$ be a sufficiently small disk centered at $\gamma(u)$ that intersects only the images of pipes incident to $u$.
(2) Subdivide every pipe $uv\in E(H)$ incident to $u$ with a new cluster $y_v$, let $\gamma(y_v):=\partial D_u\cap \gamma(uv)$.
(3) Subdivide every edge $ab\in E(G)$ such that $\lambda(b)=u$ with a new vertex $x_a$ such that $\lambda(x_a)=y_{\lambda(a)}$.
(4) For every vertex $b\in \lambda^{-1}[u]$, and any two neighbors $x_a$ and $x_c$, insert an edge $x_ax_c$ in $G$, insert a pipe $\lambda(x_a)\lambda(x_c)$ in $H$ if it is not already present, and draw this pipe in the plane as a straight-line segment between $\gamma(\lambda(x_a))$ and $\gamma(\lambda(x_c))$.
(5) Delete cluster $u$ from $H$, and delete all vertices in $\lambda^{-1}[u]$ from $G$.
(6) Return the resulting instance by $\lambda':G'\rightarrow H'$ and $\gamma':H'\rightarrow \mathbb{R}^2$.
\end{quote}

\begin{lemma}\label{lem:cluster-exp}
If $G$ is a cycle, $\lambda:G\rightarrow H$ has no spur, and $u\in V(H)$,
then \textsf{ClusterExpansion}$(u)$ produces an instance where $G'$ is a cycle, $\lambda':G'\rightarrow H'$ has no spur,
and $\crn(\gamma\circ\lambda) = \crn(\gamma'\circ\lambda')$.
\end{lemma}
We remark that $\crn(\gamma\circ\lambda)$ is invariant under the \textsf{ClusterExpansion}$(u)$ operation even in the presence of spurs, however the proof is somewhat simpler in the absence spurs, and Lemma~\ref{lem:cluster-exp} also establishes that \textsf{ClusterExpansion}$(u)$ does not create new spurs.

\begin{figure}[htbp]
\centering
\includegraphics[width=0.8\textwidth]{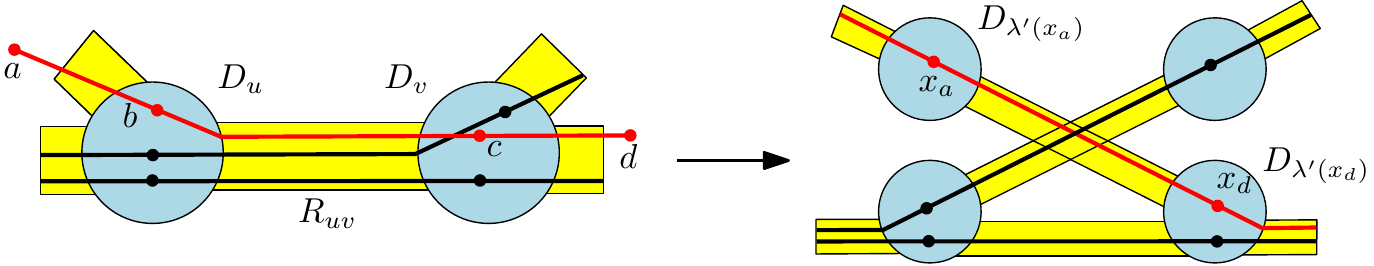}
\caption{\textsf{PipeExpansion}$(uv)$ for a safe pipe $uv$.}
\label{fig:pipeExpansion}
\end{figure}

\medskip\noindent\textbf{Pipe Expansion.}
A cluster $u\in V(H)$ is a \emph{base} of an incident pipe $uv$ if every vertex in $\lambda^{-1}[u]$ is incident to an edge in $\lambda^{-1}[uv]$. A pipe $uv\in E(H)$ is \emph{safe} if both $u$ and $v$ are bases of $uv$.
The following operation is defined on safe pipes.
See Figure~\ref{fig:expansion} for an illustration.
(We note that our algorithm would be correct even if \textsf{PipeExpansion}$(uv)$ were defined on all pipes, unlike the result in~\cite{AFT18}, since $\lambda$ does not contain spurs. We restrict this operation to safe pipe to simplify the runtime analysis.)

\begin{quote}\textsf{PipeExpansion}$(uv)$.
(1) Let $D_{uv}$ be a sufficiently narrow ellipse with foci at $\gamma(u)$ and $\gamma(v)$
that intersects only the images of pipes incident to $u$ and $v$.
(2) Subdivide every pipe $e\in E(H)$ incident to $u$ or $v$ with a new cluster $y_e$, let $\gamma(y_e):=\partial D_{uv}\cap \gamma(e)$.
(3) Subdivide every edge $ab\in E(G)$ such that $\lambda(a)\notin\{u,v\}$ and $\lambda(b)\in \{u,v\}$ with a new vertex $x_a$ such that $\lambda(x)=y_{\lambda(ab)}$.
(4) For every edge $bc\in \lambda^{-1}[uv]$, and the two neighbors $x_a$ and $x_d$ of $b$ and $c$, respectively, insert an edge $x_ax_d$ in $G$, insert a pipe $\lambda(x_a)\lambda(x_d)$ in $H$ if it is not already present, and draw this pipe in the plane as a straight-line segment between $\gamma(\lambda(x_a))$ and $\gamma(\lambda(x_d))$.
(5) Delete clusters $u$ and $v$ from $H$, and delete all vertices in $\lambda^{-1}[uv]$ from $G$.
(6) Return the resulting instance by $\lambda':G'\rightarrow H'$ and $\gamma':H'\rightarrow \mathbb{R}^2$.
\end{quote}
\begin{lemma}\label{lem:pipe-exp}
If $G$ is a cycle, $\lambda:G\rightarrow H$ has no spur, and $uv\in E(H)$ is a safe pipe,
then \textsf{PipeExpansion}$(uv)$ produces an instance where $G'$ is a cycle, $\lambda':G'\rightarrow H'$ has no spur, and $\crn(\gamma\circ\lambda) = \crn(\gamma'\circ\lambda')$.
\end{lemma}

We remark that Lemma~\ref{lem:pipe-exp} holds even for $uv$ that is not safe,
provided that $\lambda:G\rightarrow H$ has no spur.

\medskip\noindent\textbf{Main Algorithm.} Given an instance $\lambda:G\rightarrow H$ and $\gamma:H\rightarrow \mathbb{R}^2$,
we apply the two operations defined above as follows.

\smallskip

\begin{algorithm}[H]\DontPrintSemicolon
\textbf{Algorithm~1.} \KwIn{$(G, H, \lambda, \gamma)$}
 $U_0 \longleftarrow V(H)$\;
\For{every $u\in U_0$}{
    \textsf{ClusterExpansion}$(u)$\;}
\While{there is a safe pipe $uv\in E(H)$ such that $\deg_H(u)\geq 3$ or $\deg_H(v)\geq 3$}{
     \textsf{PipeExpansion}$(uv)$\;}
$uv \longleftarrow$ an arbitrary edge in $E(H)$.\;
\Return{$\crn_2(\gamma,\lambda)+|\lambda^{-1}[uv]|-1$.}
\end{algorithm}

\begin{lemma}\label{lem:term}
Algorithm~1 terminates.
\end{lemma}
\begin{proof}
By Lemmas~\ref{lem:cluster-exp} and \ref{lem:pipe-exp}, $\lambda:G\rightarrow H$ has no spurs in any step of the algorithm.
It is enough to show that the while loop of Algorithm~1 terminates. We define the potential function $\Phi(G,H)=|E(G)|-|E(H)|$,
and show that $\Phi(G,H)\geq 0$  and it decreases in every invocation of \textsf{PipeExpansion}$(uv)$.
Since $G$ is a cycle and $\lambda$ has no spur, every edge in $\lambda^{-1}[uv]$ is adjacent to one edge
in some other pipe incident to $u$ and one edge in some other pipe incident to $v$.
Each of these edges contributes to one edge in $E(G')$ inside the ellipse $D_{uv}$. Since $uv$ is safe, $G'$ has no other new edges.
Consequently, $|E(G')|=|E(G)|$. Since $\deg_H(u)\geq 3$ or $\deg_H(v)\geq 3$, \textsf{PipeExpansion}$(uv)$ replaces the clusters $u$ and $v$ with at least 3 clusters, each of which is incident to at least one pipe in the ellipse $D_{uv}$. Consequently, $|E(H')|>|E(H)|$, and so $\Phi(G,H)>\Phi(G',H')$, as claimed.
\hfill$\Box$\end{proof}

\begin{lemma}\label{lem:algo}
At the end of the while loop of Algorithm~1, $H$ is a cycle.
\end{lemma}
\begin{proof}
It is enough to show that if $H$ is not a cycle in the while loop of Algorithm~1,
then there is a safe pipe $uv\in E(H)$ such that $\deg_H(u)\geq 3$ or $\deg_H(v)\geq 3$.
Observe that every cluster created by \textsf{ClusterExpansion}$(u)$ (resp., \textsf{PipeExpansion}$(uv)$) is
a base for the unique incident pipe in the exterior of disk $D_u$ (resp., ellipse $D_{uv}$). Let $s:V(H)\rightarrow E(H)$
be a function that maps every cluster to that incident pipe. Note also that the input does not have spurs, and no spurs
are created in the algorithm by Lemmas~\ref{lem:cluster-exp} and~\ref{lem:pipe-exp}. In the absence of spurs,
if $u\in V(H)$ and $\deg_H(u)=2$, then $u$ is a base for both incident pipes.

Assume that in some step of the while loop, $H$ is not a cycle.
Let $v_1\in V(H)$ be an arbitrary cluster such that $\deg_H(v_1)\geq 3$.
Construct a maximal simple path $(v_1,v_2,\ldots, v_\ell)$ incrementally such that
$s(v_i)=v_iv_{i+1}$ for $i=1,2,\ldots \ell$. If the path encounters
a cluster $v_i$ where $s(v_i)=s(v_{i-1})$, then the pipe $v_{i-1}v_i$ is safe.
Similarly, if $\deg_H(v_{i+1})=2$, then $v_iv_{i+1}$ is safe.
Otherwise, the path ends with a repeated cluster: $s(v_\ell)=v_\ell v_i$, for some $1\leq i< \ell-1$,
and so we obtain a cycle $(v_i,v_{i+1},\ldots , v_\ell)$ of at least 3 vertices.
Let $v_j$, $i\leq j\leq \ell$, be the cluster created in the most recent
\textsf{ClusterExpansion}$(u)$ or \textsf{PipeExpansion}$(uv)$ operation.
Then $s(v_j)$ is a pipe in the exterior of a disk $D_u$ or an ellipse $D_{uv}$.
Hence, the pipe $v_{j-1}v_j$ is in the interior of $D_u$ or $D_{uv}$, moreover
$v_j$ and $v_{j-1}$ were created by the same operation. However, this
implies $s(v_{j-1})\neq v_{j-1}v_j$, contradicting the assumption that
$(v_i,v_{i+1},\ldots , v_\ell)$  is a cycle. We conclude that the path
finds a safe pipe before any cluster repeats.
\hfill$\Box$\end{proof}

\begin{lemma}\label{lem:cor}
Algorithm~1  returns $\crn(\gamma\circ \lambda)$.
\end{lemma}
\begin{proof}
By~\eqref{eq:cr1-2}, $\crn(\gamma\circ\lambda)=\crn_1(\lambda)+\crn_2(\gamma,\lambda)$.
Here $\crn_2(\gamma,\lambda)$ can be computed by a line sweep of the drawing $\gamma(H)$.
By Lemmas~\ref{lem:cycle} and \ref{lem:algo}, at the end of the algorithm, $\crn_1(\lambda)=|\lambda^{-1}[uv]|-1$
for an arbitrary edge $uv\in E(H)$. By Lemmas~\ref{lem:cluster-exp} and~\ref{lem:pipe-exp}, $\crn(\gamma\circ\lambda)$
is invariant in the operations, so the algorithm reports $\crn(\gamma\circ\lambda)$ for the input instance.
\hfill$\Box$\end{proof}

\medskip\noindent\textbf{Running Time.}
The efficient implementation of our algorithm relies on the following data structures.
For every cluster $u\in V(H)$ we maintain the set of vertices of $V(G)$ in $\lambda^{-1}[u]$.
For every pipe $uv\in E(H)$, we maintain $\lambda^{-1}[uv]\subset E(G)$, the weight
$w(uv)=|\lambda^{-1}[uv]|$, and the sum of weights of all pipes that cross $uv$, that we denote by $W(uv)$.
Then we have $\crn_2(\gamma,\lambda)=\frac12 \sum_{uv\in E(H)} w(uv)W(uv)$. We also maintain the
current value of $\crn_2(\gamma,\lambda)$.
We further maintain indicator variables that support checking the conditions of the while loop in
Algorithm~1: (i) whether the cluster is a base for the pipe,
(ii) whether a cluster has degree 2, and (iii) whether a pipe is safe.

\begin{lemma}\label{lem:time}
With the above data structures, Algorithm~1 runs in $O((M+R)\log M)$ time,
where $M=|E(H)|+|E(G)|$ and $R=\crn(\gamma\circ\lambda)<M^2$.
\end{lemma}

\section{NP-Completeness in the Presence of Spurs}
\label{sec:hardness}

In this section, we prove Theorem~\ref{thm:hardness}. In a problem instance, we are given
a simplicial map $\lambda:G\rightarrow H$, a straight-line drawing $\gamma:H\rightarrow \mathbb{R}^2$,
and a nonnegative integer $K$, and ask whether $\crn(\gamma\circ \lambda)\leq K$.

\begin{lemma}\label{lem:NP}
The above problem is in NP.
\end{lemma}
\begin{proof}
A feasible drawing $\Gamma\circ \Lambda:G\rightarrow \mathbb{R}^2$ with $\crn(\Gamma\circ\Lambda)\leq K$ can be witnessed by a combinatorial representation of $\Lambda$. Specifically, we can determine $\crn_2(\gamma,\lambda)$ by computing the weight of each pipe $uv\in E(H)$ in $O(|E(G)|+|E(H)|)$ time, and finding all edge-crossings in the drawing $\gamma(H)$ in $O(|E(H)|\log |E(H)|)$ time.
Given a combinatorial representation of a drawing $\Lambda:G\rightarrow \mathcal{H}$,
we can determine the number of crossings at all nodes $u\in V(H)$ in $O(\sum_{u\in V(H)}|\lambda^{-1}[u]|)=O(|E(G)|)$ time.
\hfill$\Box$\end{proof}

We prove NP-hardness by a reduction from 3SAT. Let $\Phi$ be a boolean formula in 3CNF with a set $\mathcal{X}=\{x_1,\ldots , x_n\}$ of variables and a set $\mathcal{C}=\{c_1,\ldots, c_m\}$ of clauses. We construct graphs $G$ and $H$, a simplicial map $\lambda:G\rightarrow H$, a straight-line drawing $\gamma:H\rightarrow \mathbb{R}^2$, and an integer $K\in \mathbb{N}$ such that $\crn(\gamma\circ \lambda)\leq K$ if and only if $\Phi$ is satisfiable.

\medskip\noindent\textbf{First Construction: Disjoint Union of Paths.}
Refer to Fig.~\ref{fig:variable}.

\begin{figure}[htbp]
	\centering
	\includegraphics[width=0.9\textwidth]{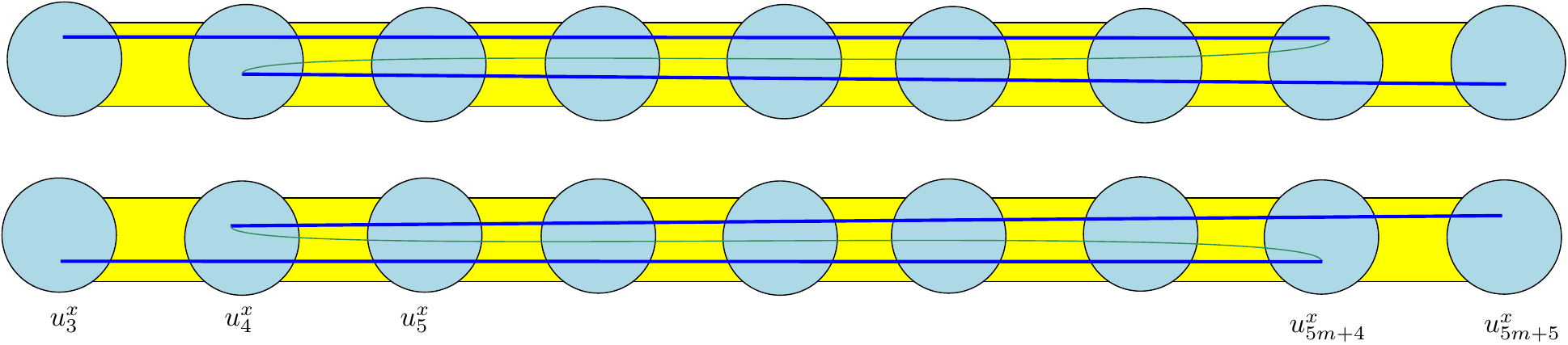}
	\caption{\textsf{Two embeddings of $G_x$. 
        Top: $P^x_1$ is above $P^x_3$. Bottom: $P^x_1$ is below $P^x_3$.}
	\label{fig:variable}}
\end{figure}

\noindent\textbf{Construction of $H$ and $\gamma: H\rightarrow \mathbb{R}^2$.}
For every variable $x\in \mathcal{X}$, create a path
$H_x=(u^x_3,u^x_4,\ldots , u^x_{5m+5})$.

For $i=1,\ldots, m$, the $i$-th clause $c_i\in \mathcal{C}$ is associated to at most three
(negated or non-negated) variables, say, $x,y,z\in \mathcal{X}$. Identify the clusters $u^x_{5i+\ell}=u^y_{5i+\ell}=u^z_{5i+\ell}$ for $\ell=0,1,2,3$ and we denote the resulting clusters also by $u_{5i+\ell}$ and associate them with clause $c_i$.
Add two new clusters $v_i$ an $w_i$, and two new pipes $v_iu^x_{5i+1}$ and $w_iu^x_{5i+2}$.
This completes the description of $H$.

For every $i=1,\ldots , m$, we map clusters $u_{5i},\ldots ,u_{5i+3}$
 to integer points
$5i,\ldots , 5i+3$ on the $x$-axis.
The two additional clusters, $v_i$ and $w_i$, are mapped to points
$\gamma(v_i)=(5i+1,1)$ and $\gamma(w_i)=(5i+2,-1)$, above and below the $x$-axis.
The remaining clusters and pipes of $H_x$, $x\in \mathcal{X}$,
are mapped to integer points in the horizonal line $y=j+1$.
Specifically, $\gamma(u_i^{x_j})=(i,j+1)$, for $3\leq i\leq 5m+5$,
except for clusters $u_i^{x_j}$ that have been merged and incorporated
in clause gadgets.

\begin{observation}\label{obs:mon}
For every $x\in \mathcal{X}$, $\gamma(H_x)$ is an $x$-monotone polygonal path in the plane.
This ensures, in particular, that if $c_i\in \mathcal{C}$ contains variables $x$, $y$, and $z$,
then the pipes of $H_x$, $H_y$, and $H_z$ that enter $u_{5i}$ and exit $u_{5i+3}$
appear in reverse ccw order in the rotation of $u_{5i}$ and $u_{5i+3}$, respectively.
\end{observation}

\noindent\textbf{Construction of $G$ and $\lambda:G\rightarrow H$.}
For each clause $c_i\in \mathcal{C}$, create a path $G_i$ of 4 vertices mapped to
$(v_i,u_{5i+1},u_{5i+2},w_i)$.
For each variable $x\in \mathcal{X}$, create a path $G_x$ as follows.
First create a path of $15m+5$ vertices as a concatenation of three paths:
$P^x_1$, $P^x_2$, and $P^x_3$, which are mapped to
$(u^x_3,\ldots , u^x_{5m+4})$,
$(u^x_{5m+4},\ldots, u^x_4)$, and
$(u^x_4,\ldots, u^x_{5m+5})$, respectively.
We shall modify $P^x_1$ and $P^x_3$ within each cluster.
Regardless of these local modifications,
in every embedding of $G_x$, the path $P^x_2$
lies between $P^x_1$ and $P^x_3$. The truth value of
variable $x$ is encoded by the above-below relationship
between $P^x_1$ and $P^x_3$ (Fig.~\ref{fig:variable}(a-b)).

Each pair $(x,c_i)\in \mathcal{X}\times \mathcal{C}$, where
a literal $x$ or $\overline{x}$ appears in $c_i$,
corresponds to the subpath $(u_{5i},\ldots ,u_{5i+3})$ of $H_x$. Suppose that a subpath $A\subset P_1^x$ and $B\subset P_3^x$
are mapped to this subpath. To simplify notation, we assume that $A$ and $B$
are \emph{directed} from $u_{5i}$ to $u_{5i+3}$.

Refer to Fig.~\ref{fig:clause}.
If $c_0$ contains the non-negated $x$,
then replace $A$ on $P_x^1$ with a subpath mapped to
$A'=(u_{5i},u_{5i+1},u_{5i+2},u_{5i+3},u_{5i+2},u_{5i+1},u_{5i+2},u_{5i+3})$ and $B$
with a subpath mapped to
$B'=(u_{5i},u_{5i+1},u_{5i+2},u_{5i+1},u_{5i},u_{5i+1},u_{5i+2},u_{5i+3})$.
If $c_0$ contains the negated $\overline{x}$
then replace $A$ with $B'$, and $B$ with $A'$.
This completes the definition of $G$.

The drawing $\gamma:H\rightarrow \mathbb{R}^2$ and $\lambda:G\rightarrow H$ determine
$\crn_2(\gamma,\lambda)$. Let $K=\crn_2(\gamma,\lambda)+13m$.
Note that $G$ and $H$ have $O(mn)$ vertices and edges,
and the drawing $\gamma$ maps the clusters in $V(H)$ to
integer points in an $O(m)\times O(n)$ grid.

\begin{figure}[htbp]
	\centering
	\includegraphics[width=0.8\textwidth]{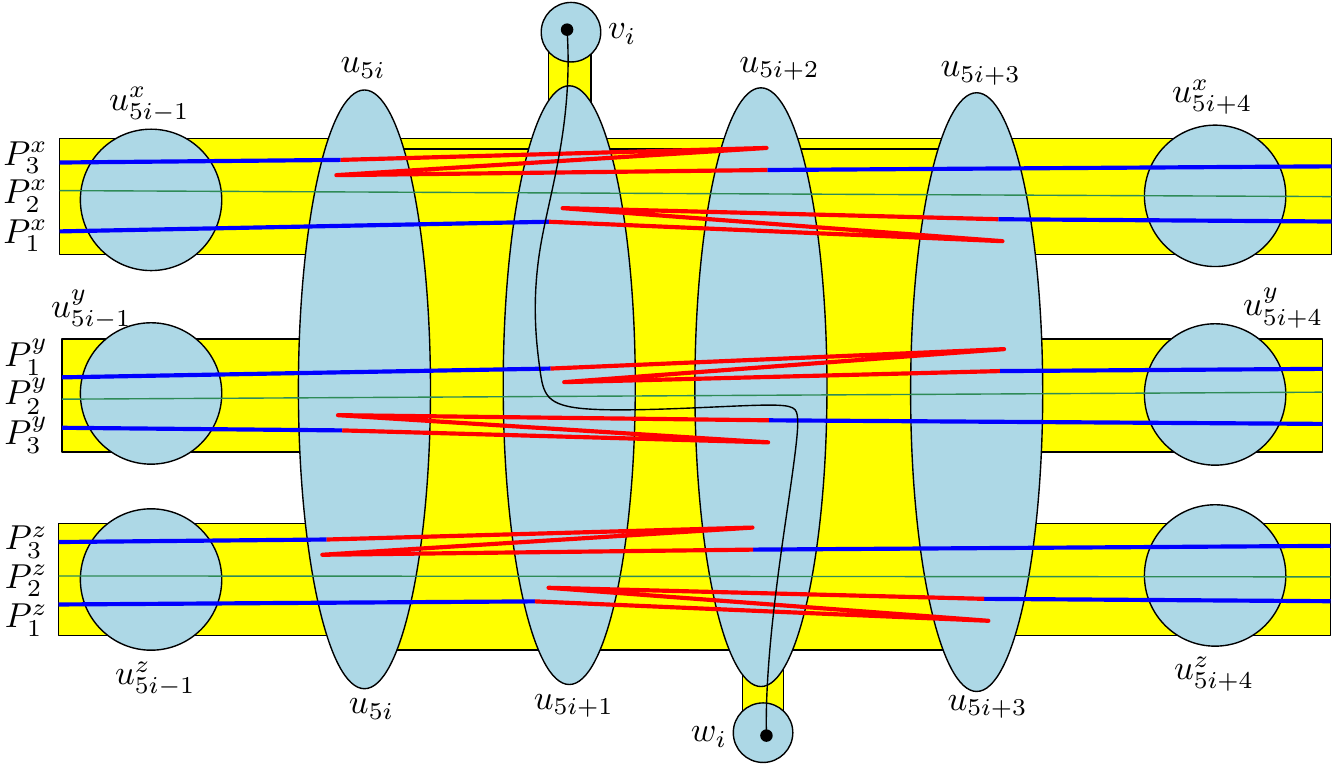}
	\caption{\textsf{A clause gadget for $c_i=(x\vee y\vee z)$, where
    $\tau(x)=\tau(z)=\text{false}$ and $\tau(y)=\text{true}$. The neighborhood of the four middle ``vertically prolonged'' clusters and pipes between them forms $\mathcal{N}_i$.}
	\label{fig:clause}}
\end{figure}

\noindent\textbf{Equivalence.} First, we show that the  satisfiability of $\Phi$ implies that $\crn(\gamma,\lambda)\le K$.
Assume that $\Phi$ is satisfiable, and let $\tau:\mathcal{X}\rightarrow \{\text{true},\text{false}\}$ be a satisfying truth assignment.
Fix $\varepsilon\in (0,\varepsilon_0)$. For every $x\in \mathcal{X}$, denote by $\mathcal{N}_x$ the union of disks $N_u$ and $N_{uv}$ for all clusters $v\in V(H_x)$ and pipes $uv\in E(H_x)$; and similarly let $\mathcal{N}_i$ be the union of such regions for the path $(u_{5i},\ldots , u_{5i+3})$ in $H$. For every $x\in \mathcal{X}$, incrementally, embed the path $G_x$ in $\mathcal{N}_x$ as follows: each edge is an $x$-monotone Jordan arc; if $\tau(x)=\text{true}$, then $P^x_1$ lies above $P^x_3$; otherwise $P^x_3$ lies above $P^x_1$. If a clause $c_i$ contains variables $x,y,z\in \mathcal{X}$, we also ensure that the embeddings of $G_z$, $G_y$, and $G_y$ are pairwise disjoint within $\mathcal{N}_i$. This is possible by Observation~\ref{obs:mon}. Finally, for $i=1,\ldots , m$, embed the path $G_i$ as follows. Assume that $c_i$ contains the variables $x,y,z\in \mathcal{X}$, where $x$ corresponds to a true literal in $c_i$. Then $\Gamma(G_i)$ starts from $\gamma(v_i)$ along the vertical line $x=5i+1$ until it crosses the arc $\Gamma(P^x_2)$, then follows $\Gamma(P^x_2)$ to the vertical line $x=5i+2$, and continues to $\gamma(w_i)$ along that line. Note that $\Gamma(P^x_2)$ crosses only 3 edges in $\Gamma(G_x)$, and 5 edges in $\Gamma(G_y)$ and $\Gamma(G_z)$. So there are 13 crossings in $\mathcal{N}_i$ for $i=1,\ldots , m$; and the total number of crossings is $\crn_2(\gamma,\lambda)+13m$, as required. \\

Second, we show  that $\crn(\gamma,\lambda)\le K$ implies that  $\Phi$ is satisfiable by constructing a satisfying assignment. Consider functions $\Lambda:G\rightarrow \mathcal{H}$ and $\Gamma:\mathcal{H}\rightarrow \mathbb{R}^2$ such that $\Gamma\circ\Lambda:G\rightarrow \mathbb{R}^2$ is a drawing in which $\crn(\Gamma\circ \Lambda)\leq K$. Note that $\crn_2(\gamma,\lambda)$ crossings are unavoidable due to edge-crossings in the drawing $\gamma(H)$. Hence, by the definition of $K$, there are at most $13m$ crossings in the neighborhoods of clusters. We show that (1) there must be precisely 13 crossings in each neighborhood $\mathcal{N}_i$, (2) $\Gamma\circ\Lambda(G_x)$ is an embedding for every $x\in \mathcal{X}$, and (3) the embeddings of $G_x$, for all $x\in \mathcal{X}$, jointly encode a satisfying truth assignment for $\Phi$. (1) and (2) is established by the following lemma.

\begin{lemma}\label{lem:13}
Let $i\in \{1,\ldots , m\}$ and let $x,y,z\in \mathcal{X}$ be the three variables in $c_i$. In $\Gamma\circ\Lambda$, there are at least 13 crossings in neighborhood $\mathcal{N}_i$,
and equality is possible only if none of the drawings $\Gamma\circ\Lambda(G_x)$, $x\in \mathcal{X}$, has self-crossings in $\mathcal{N}_i$, and  at least one of $G_x, G_y$ and $G_z$ is crossed exactly 3 times by $G_i$.
\end{lemma}

By Lemma~\ref{lem:13}, $\crn_1(\lambda)\leq 13m$ implies that $\Gamma\circ \Lambda$ defines an embedding of $G_x$, for all $x\in \mathcal{X}$, in each region $\mathcal{N}_i$, $i=1,\ldots, m$. Consequently, $\Gamma\circ \Lambda$ defines an embedding of $G_x$ in $\mathbb{R}^2$ for all $x\in \mathcal{X}$. In every embedding  $\Gamma\circ \Lambda(G_x)$, for $x\in \mathcal{X}$, either $P^x_1$ lies above $P^x_2$, or vice~versa. We can now
define a truth assignment $\tau:\mathcal{X}\rightarrow \{\text{true},\text{false}\}$ such that for every $x\in \mathcal{X}$, $\tau(x)=\text{true}$ if and only if $P^x_1$ lies above $P^x_2$ in $\Gamma\circ \Lambda(G_x)$.

\begin{lemma}\label{lem:sat}
Assume that $\Gamma\circ\Lambda(G_x)$ is an embedding for every $x\in \mathcal{X}$, which determines the truth assignment $\tau:\mathcal{X}\rightarrow \{\text{true},\text{false}\}$ described above. For every $i=1,\ldots , m$, if variable $x$ appears in clause $c_i$, and $G_i$ crosses $G_x$ at most 3 times in $\mathcal{N}_i$, then $x$ appears as a true literal in $c_i$.
\end{lemma}
\begin{proof}
Consider the highest and lowest path $P_h$ and $P_l$ among $P^x_1,P^x_2$ or $P^x_3$, respectively, in
$\mathcal{N}_i\cap \Gamma\circ\Lambda(G_x)$, none of which can be $P^x_2$ since $\Gamma\circ\Lambda(G_x)$ is an embedding.
By the construction of $\lambda$, either there exists exactly one pipe-degree 2 component of $P_h$  in $\lambda^{-1}[u_{5i+1}]$ and exactly one pipe-degree 2 component of $P_l$ in $\lambda^{-1}[u_{5i+2}]$, or vice~versa.

 By the construction of $\lambda$, $G_i$ crosses each of $P^x_1$, $P^x_2$, and $P^x_3$ at least once in $\mathcal{N}_i$. By the hypothesis of the lemma,  it crosses each exactly once. Then $P_h$ has only one pipe-degree 2 component in $\lambda^{-1}[u_{5i+1}]$, and  $P_l$ has only one pipe-degree 2 component in $\lambda^{-1}[u_{5i+2}]$.
By the construction of $\lambda$, if $x$ appears as a non-negated literal in $c_i$ this means that $P_h=P_1^x$ lies above $P_2^x$ and therefore $\tau(x)=\text{true}$.
Similarly, if $x$ appears as a negated literal in $c_i$ this means that $P_3^x=P_h$ lies above $P_2^x$ and therefore $\tau(x)=\text{false}$.
Consequently, $x$ appears as a true literal in $c_i$ and that concludes the proof.
\hfill$\Box$\end{proof}

Since $\crn_1(\lambda)\leq 13m$,  for every $i=1,\ldots ,m$,
there are exactly 13 crossings in $\mathcal{N}_i$  by Lemma~\ref{lem:13}. Moreover, by Lemma~\ref{lem:13} the drawing $\Gamma\circ\Lambda(G_x)$ is an embedding for every $x\in \mathcal{X}$,
and in every $c_i$ for one its variables $x$ the drawing of $G_x$
is crossed by $G_i$ exactly 3 times.
By Lemma~\ref{lem:sat},  the assignment $\tau$ makes at least one literal in each clause $c_i$ of $\Phi$ true.
We conclude that $\Phi$ is satisfiable, as required.
This completes the proof of NP-hardness.

\medskip\noindent\textbf{Second Construction: Cycle.}
In our first construction, $G$ is a disjoint union of paths, and for every path endpoint $a\in V(G)$,
$a$ is the only vertex mapped to the cluster $\lambda(a)\in V(H)$. This property allows us to expand the construction as follows.
We augment $G$ into a cycle $\overline{G}$ by adding a perfect matching $M_G$ connecting the path endpoints, and we augment $H$ with the
corresponding matching between the clusters $M_H=\{\lambda(a)\lambda(b):ab\in M_G\}$, and for every new pipe $uv\in M_H$
draw a polygonal arc $\gamma(uv)$ between $\gamma(u)$ and $\gamma(v)$ that does not pass through the image of any other cluster
(but may cross images of other pipes). The augmentation does not change $\crn_1(\lambda)$, and we can easily
compute the increase in $\crn_2(\gamma,\lambda)$ due to new crossings. Consequently, finding $\crn(\gamma\circ \lambda)$ remains NP-hard.

\section{Conclusions}
\label{sec:con}

Motivated by recent efficient algorithms that can decide whether a piecewise linear map $\varphi:G\rightarrow \mathbb{R}^2$ can be perturbed into an embedding, we investigate the problem of computing the minimum number of crossings in a perturbation. We have described an efficient algorithm when $G$ is a cycle and $\phi$ has no spurs (Theorem~\ref{thm:cycle}); and the problem becomes NP-hard if $G$ is an arbitrary graph, or if $G$ is a cycle but $\varphi$ may have spurs (Theorem~\ref{thm:hardness}). However, perhaps one can minimize the number of crossings efficiently under milder assumptions. We formulate one promising scenario as follows: Is there a polynomial-time algorithm that finds $\crn(\gamma\circ\lambda)$ when $\lambda^{-1}[u]$ is a planar graph (resp., an edgeless graph) for every cluster $u\in V(H)$ and $\lambda$ has no spurs?

\appendix
\section{Omitted Proofs}

\rephrase{Lemma}{\ref{lem:cluster-exp}}{If $G$ is a cycle, $\lambda:G\rightarrow H$ has no spur, and $u\in V(H)$,
then \textsf{ClusterExpansion}$(u)$ produces an instance where $G'$ is a cycle, $\lambda':G'\rightarrow H'$ has no spur,
and $\crn(\gamma\circ\lambda) = \crn(\gamma'\circ\lambda')$.}

\begin{proof}
If $G$ is a cycle, then every vertex $b\in \lambda^{-1}[u]$ has precisely two neighbors, say $a$ and $c$. Step~3 subdivides these edges with new vertices $x_a$ and $x_c$, Step~4 inserts an edge $x_ax_c$, and Step~6 deletes $b$. Consequently, the path $(a,b,c)$ is replaced by a path $(a,x_a,x_c,c)$. Since all such paths are edge-disjoint, the resulting graph $G'$ is a cycle.

Since $\lambda:G\rightarrow H$ has no spur, for every vertex $b\in \lambda^{-1}[u]$, the neighbors $a$ and $c$ are in distinct clusters, that is $\lambda(a)\neq \lambda(c)$. Consequently, $y_{\lambda(a)}\neq y_{\lambda(c)}$ and so $\lambda'(x_a)\neq \lambda'(x_c)$. Therefore the operation does not create spurs.

Let $\Lambda:G\rightarrow \mathcal{H}$ be a drawing that attains $\crn_1(\lambda)$. We may assume that every connected component of $\Lambda(G)\cap D_u$ and $\Lambda(G)\cap R_{uv}$ is a straight-line segment.

    Let $(a,b,c)$ and $(d,e,f)$ be two different paths in $G$ such that $\lambda(b)=\lambda(e)=u$.
There are two types of crossings of $\Lambda$ in $D_u$ between paths $(a,b,c)$ and $(d,e,f)$ as above.  In the  first type,  $\lambda(a)$ and $\lambda(c)$ interleave in the rotation at $u$ with $\lambda(d)$ and $\lambda(f)$.
 In the  second type, we have $\lambda(a)=\lambda(d)$, $\lambda(a)=\lambda(f)$, $\lambda(c)=\lambda(d)$, or $\lambda(c)=\lambda(f)$.

 Let $\text{CR}_{\Lambda}^{\times}(u)$ denote the number of crossings of the first type.    Let $\text{CR}_{\Lambda}^{<}(u)$ denote the number of crossings of the second type.
 In the following we construct $\Lambda':G'\rightarrow \mathcal{H}'$  witnessing $\crn(\gamma'\circ\lambda')\le\crn(\gamma\circ\lambda)$  such that
$\left(\sum_{u\in V(H')} \text{CR}_{\Lambda'}(u)\right)=\crn_1(\lambda)-\text{CR}_{\Lambda}^{\times}(u)$ and  $   \crn_2(\gamma',\lambda') = \crn_2(\lambda)+\text{CR}_{\Lambda}^{\times}(u)$.
Note that the second condition does not depend on $\Lambda'$ and follows by the construction of $\gamma'$.

Let $h$ denote the natural homeomorphism between $\mathcal{H}\setminus \text{int}(D_u)$ and the connected component of $\mathcal{H}'\setminus \bigcup_{uv\in E(G)}\text{int}(D_{y_v})$ containing $D_v$'s for $v\not=u$. Thus, $h$ takes $D_v$'s of $\mathcal{H}$ to $D_v$'s of $\mathcal{H}'$, and similarly $R_{vw}$'s of $\mathcal{H}$  to $R_{vw}$'s of $\mathcal{H}'$.
We put $\Lambda'(vw)=h(\Lambda(vw))$, if $\lambda(v),\lambda(w)\not=u$.
We define $\Lambda'$ on every path $(a,x_a,x_c,c)$, that replaced in $G'$  path $(a,b,c)$ in $G$ such that $\lambda(b)=u$, as follows.  Let $p_{ab}= \partial(\mathcal{H}\setminus \text{int}(D_u)) \cap \Lambda(ab)$
and $p_{bc}= \partial(\mathcal{H}\setminus \text{int}(D_u)) \cap \Lambda(bc)$. We define $\Lambda'(a,x_a)$ as the concatenation of the polygonal line from $h(\Lambda(a))$ to $h(p_{ab})$ contained in $\Lambda(ab)$ and a very short crossing free line segment contained in $D_{\lambda'{(x_a)}}$. In the same manner we construct $\Lambda'(x_c,c)$.
Let $(a',x_{a'},x_{c'},c')$ denote another such path, i.e., $(a',x_{a'},x_{c'},c')$ replaced $(a',b',c')$ in $G$ such that $\lambda(b')=u$.

We construct  $\Lambda'(x_a,x_c)$ as a polygonal line with at most two bends at  $\partial D_{\lambda(x_a)}$ and  $\partial D_{\lambda(x_c)}$ so that $\Lambda'(x_a,x_c)$
and $\Lambda'(x_{a'},x_{c'})$ cross if and only if $p_{ab}$ and $p_{bc}$ interleave with $p_{a'b'}$ and $p_{b'c'}$ along  $\partial D_u$, and $\{\lambda'(x_a),\lambda'(x_c)\}\cap \{\lambda'(x_{a'}),\lambda'(x_{c'})\}\not=\emptyset$. In the case when
 $\Lambda'(x_a,x_c)$
and $\Lambda'(x_{a'},x_{c'})$ cross, we also require that they cross exactly once.
Let $p_{ac}^a$ and $p_{ac}^c$ denote the intersection of $\Lambda'(x_{a},x_{c})$ with $\partial D_{\lambda(x_a)}$ and  $\partial D_{\lambda(x_c)}$,
respectively. It is enough to specify $\Lambda'$ by presenting constraints on the order of the intersection points of the edges $x_ax_c$ with $\partial D_{\lambda(x_a)}$ and $\partial D_{\lambda(x_c)}$ enforcing the previously mentioned property, and realize these constraints by the corresponding cyclic orders of these points.

Let us fix a total order $<$ on the vertices of $V(H')\setminus V(H)$.
If either $\lambda'(x_a)= \lambda'(x_{a'})$ and $\lambda'(x_c)= \lambda'(x_{c'})$ and  $\lambda'(x_a)< \lambda'(x_c)$; or $\lambda'(x_a)= \lambda'(x_{a'})$ and  $\lambda'(x_c)\not= \lambda'(x_{c'})$, we construct $\Lambda'(x_{a},x_{c})$ so that  $h(p_{ab})$ and $p_{ac}^a$ interleave with $h(p_{a'b'})$ and $p_{a'c'}^{a'}$ along $\partial D_{\lambda(x_a)}$ if and only
if  $p_{ab}$ and $p_{bc}$ interleave with $p_{a'b'}$ and $p_{b'c'}$  along  $\partial D_u$. We make the points
$p_{ac}^c$ and $p_{bc}$  not to interleave with $p_{a'c'}^{c'}$ and $p_{b'c'}$.
Clearly, the given constraints yield the desired properties.

These constraints are realized by an inductive construction using  the
total order $<$ on the vertices of $H'$.
First, let $u'$ be the first vertex in this order.
The order of $h(p_{ab})$'s and $p_{ac}^a$'s along the boundary of $D_{u'}$ is obtained from the order of $p_{ab}$'s and $p_{bc}$'s along $D_u$ via the bijection $h(p_{ab})\leftrightarrow p_{ab}$ and $p_{ac}^a\leftrightarrow p_{bc}$.
Throughout the induction we maintain the invariant that for every $u'\in V(H')\setminus V(H)$ the cyclic order of considered $h(p_{ab})$'s along  $\partial D_{u'}$ is the same as the cyclic order of the corresponding $p_{ab}$'s along  $\partial D_{u}$ via the above bijection, which clearly holds after the base step.
In the inductive step we consider $u'\in V(H')\setminus V(H)$ and we need to specify the orders for the $p_{ac}^c$'s such that $\lambda(x_c)>\lambda(x_a)=u'$. This we do analogously as in the base step, and due to
the invariant the inductive step goes through.
This concludes the proof for
 $\crn(\gamma'\circ\lambda')\le\crn(\gamma\circ\lambda)$.

 To establish the other direction, we  start with a drawing $\Lambda':G'\rightarrow \mathcal{H}'$ witnessing $\crn(\gamma'\circ\lambda')$ apply the inverse of $h$ to construct $\Lambda$ in $\mathcal{H}\setminus D_u$. Finally, it is enough to observe that the order of intersection points $p_{ab}$ along $\partial D_u$ specifies $\lambda$ for which
$\left(\sum_{u\in V(H)} \text{CR}_{\Lambda}(u)\right)\le\crn_1(\lambda')-\text{CR}_{\Lambda}^{\times}(u)$ and  $   \crn_2(\gamma',\lambda') = \crn_2(\lambda)+\text{CR}_{\Lambda}^{\times}(u)$, and that concludes the proof.\hfill$\Box$
\end{proof}

\rephrase{Lemma}{\ref{lem:pipe-exp}}{If $G$ is a cycle, $\lambda:G\rightarrow H$ has no spur, and $uv\in E(H)$ is a safe pipe,
then \textsf{PipeExpansion}$(uv)$ produces an instance where $G'$ is a cycle, $\lambda':G'\rightarrow H'$ has no spur,
and $\crn(\gamma\circ\lambda) = \crn(\gamma'\circ\lambda')$.
}

\begin{proof}
The proof is the almost the same as  the proof of Lemma~\ref{lem:cluster-exp}
with $D_u \cup R_{uv} \cup D_v$ playing the role of $D_u$.
Instead of paths  $(a,b,c)$ such that $\lambda(b)=\lambda(e)=u$
we consider paths $(a,b,c,d)$ such that $\lambda(b)=u$ and $\lambda(c)=v$.

If $G$ is a cycle, then the both vertices of every $bc\in \lambda^{-1}[uv]$ have precisely one other neighbor, say $a$ for $b$ and $d$ for $c$. Step~3 subdivides these edges with new vertices $x_a$ and $x_d$, Step~4 inserts an edge $x_ax_d$, and Step~6 deletes $b$ and $c$. Consequently, the path $(a,b,c,d)$ is replaced by a path $(a,x_a,x_d,d)$. Since all such paths are edge-disjoint, the resulting graph $G'$ is a cycle. Clearly, the operation of \textsf{PipeExpansion} does not create spurs, since $\lambda$ has no spur.

Let $\Lambda:G\rightarrow \mathcal{H}$ be a drawing that attains $\crn_1(\lambda)$. We may assume that every connected component of $\Lambda(G)\cap D_u$, $\Lambda(G)\cap D_v$ and $\Lambda(G)\cap R_{uv}$ is a straight-line segment.

  There are two types of crossings of $\Lambda(G) \cap (D_u \cup D_v \cup R_{uv})$.  Let $(a,b,c,d)$ and $(e,f,g,h)$ be paths in $G$ such that $\lambda(b)=\lambda(f)=u$ and $\lambda(c)=\lambda(g)=v$.
 In the  first type,  $\lambda(a)$ and $\lambda(d)$ interleave in the rotation at $uv$ with $\lambda(e)$ and $\lambda(h)$.
 In the  second type, we have $\lambda(a)=\lambda(e)$, $\lambda(a)=\lambda(h)$, $\lambda(d)=\lambda(e)$, or $\lambda(d)=\lambda(h)$.

 Let $\text{CR}_{\Lambda}^{\times}(uv)$ be the number of crossings of the first type.    Let $\text{CR}_{\Lambda}^{<}(uv)$ denote the number of crossings of the second type.
 Analogously as in the proof of Lemma~\ref{lem:cluster-exp}
 with $D_u\cup R_{uv} \cup D_v$ playing the role of $D_u$, we construct $\Lambda':G'\rightarrow \mathcal{H}'$  witnessing $\crn(\gamma'\circ\lambda')\le\crn(\gamma\circ\lambda)$  such that
$\left(\sum_{u\in V(H')} \text{CR}_{\Lambda'}(u)\right)=\crn_1(\lambda)-\text{CR}_{\Lambda}^{\times}(uv)$ and  $   \crn_2(\gamma',\lambda') = \crn_2(\lambda)+\text{CR}_{\Lambda}^{\times}(uv)$.
Note that the second condition does not depend on $\Lambda'$ and follows by the construction of $\gamma'$.

 To establish the other direction, we can start with a drawing $\Lambda':G'\rightarrow \mathcal{H}'$ witnessing $\crn(\gamma'\circ\lambda')$ apply the inverse of analog of $h$ from the proof of Lemma~\ref{lem:cluster-exp} to construct $\Lambda$ in $\mathcal{H}\setminus (D_u\cup R_{uv} \cup D_v)$. Finally, it is enough to observe that the order of intersection points $p_{ab}$ along $\partial( D_u \cup R_{uv} \cup D_v)$ yields $\lambda$ for which
$\left(\sum_{u\in V(H)} \text{CR}_{\Lambda}(u)\right)\le\crn_1(\lambda')-\text{CR}_{\Lambda}^{\times}(uv)$ and  $   \crn_2(\gamma',\lambda') = \crn_2(\lambda)+\text{CR}_{\Lambda}^{\times}(uv)$. Here, we can again treat $(D_u\cup R_{uv} \cup D_v)$ as $D_u$ in the proof of Lemma~\ref{lem:cluster-exp}, and that concludes the proof.
\hfill$\Box$\end{proof}

\rephrase{Lemma}{\ref{lem:time}}{
With the above data structures, Algorithm~1 runs in $O((M+R)\log M)$ time,
where $M=|E(H)|+|E(G)|$ and $R=\crn(\gamma\circ\lambda)<M^2$.}
\begin{proof}
At preprocessing, we can compute $\lambda^{-1}[u]$, $\lambda^{-1}[uv]$, and $w(uv)$ by a simple traversal of $G$ in $O(|E(G)|)$ time. Since every crossing in the drawing $\gamma(H)$ corresponds to at least one crossing in any perturbation, $\gamma(H)$ has at most $R$ crossings. Hence the complexity of the arrangement of all edges in $\gamma(H)$ is $O(M+R)$.
A standard line sweep algorithm can find all crossings of $\gamma(H)$ in $O((M+R)\log (M+R))=O((M+R)\log M)$
time. The same algorithm can also compute $W(uv)$ for all $uv\in E(H)$, and $\crn_2(\gamma,\lambda)$.

Algorithm~1 starts with a for-loop over all $u\in U_0$. We can update $\lambda^{-1}[u]$, $\lambda^{-1}[uv]$, and $w(uv)$ in
$O(\deg_H(u)+|\lambda^{-1}(u)|)$ time per \textsf{ClusterExpansion}$(u)$. This sums to $O(|E(H)|+|E(G)|)$ time for all $u\in U_0$.
All new crossings in $\gamma(H)$ occur between the pipes created in the interior of the disks $D_u$, for all $u\in U_0$.
These crossings can be found in $O((M+R)\log M)$ total time.

Note also that \textsf{ClusterExpansion}$(u)$, for all $u\in U_0$ doubles the number of edges in $G$. However, $|E(G)|$ is invariant under  \textsf{PipeExpansion} operations. In fact, \textsf{PipeExpansion}$(uv)$ partitions the set $\lambda^{-1}[uv]\subset E(G)$ into two or more subsets, which are mapped to pipes in the ellipse $D_{uv}$, and the $\lambda^{-1}(e)$ for every other pipe $e\in E(H)$ remains unchanged.
We maintain $\lambda^{-1}[u]$, $\lambda^{-1}[v]$, $\lambda^{-1}[uv]$, and $w(uv)$ in the while loop of Algorithm~1 using a heavy-path decomposition.
Suppose \textsf{PipeExpansion}$(uv)$ replaces $uv$ with pipes $u_1v_1,\ldots , u_kv_k$, which correspond to pairs of clusters in the neighborhood of
$u$ and $v$, respectively. The naive implementation would take $O(w(uv))$ time, but we can reduce it to $O(w(uv)-\max_{i} w(u_iv_i))$:
Put $S=\lambda^{-1}[uv]$ and compute the sets $\lambda^{-1}[u_iv_i]$ incrementally in parallel by deleting edges from $S$; when all but maximal  set has been computed, then all remaining elements of $S$ can be added to this maximal set in $O(1)$ time. The time $O(w(uv)-\max_{i} w(u_iv_i))$
can then be charged to the edges that move from $\lambda^{-1}[uv]$ to a set $\lambda^{-1}[u_iv_i]$ with $w(u_iv_i)\leq w(uv)/2$. Over all operations of the while loop of Algorithm~1, edges that are initially mapped to a pipe of weight $w$ receive a charge of at most $O(\sum_{i=0}^\infty 2^i\lfloor w/2^i\rfloor)=O(w\log w)$. Summation over all edges of $E(H)$ yields $O(\sum_{uv\in E(H)} w(uv)\log w(uv)) \leq O(|E(G)|\log |E(G)|)=O(M\log M)$.

Also, \textsf{PipeExpansion}$(uv)$ replaces $uv$ with pipes $u_1v_1,\ldots , u_kv_k$, then every pipe that crossed $uv$ will cross $u_1v_1,\ldots , u_kv_k$. So $W(u_iv_i)$, $i=1,\ldots ,k$, can be computed by adding the number of \emph{new} crossings to $W(uv)$. All new crossings created by \textsf{PipeExpansion}$(uv)$ are between new pipes in the ellipse $D_{uv}$. Since pipe crossings are never removed, the total number of such pipe crossings is at most $R$, and they can be computed in $O((M+R)\log M)$ time over all operations of the while loop of Algorithm~1.

At the end of the algorithm, both $\crn_1(\lambda)=w(uv)-1$ for an arbitrary pipe $uv\in E(H)$,
and $\crn_2(\gamma,\lambda)=\frac12 \sum_{uv\in E(H)} w(uv)W(uv)$ can be calculated in $O(M)$ time.
\hfill$\Box$\end{proof}

\rephrase{Lemma}{\ref{lem:13}}{
Let $i\in \{1,\ldots , m\}$ and let $x,y,z\in \mathcal{X}$ be the three variables in $c_i$. In $\Gamma\circ\Lambda$, there are at least 13 crossings in neighborhood $\mathcal{N}_i$,
and equality is possible only if none of the drawings $\Gamma\circ\Lambda(G_x)$, $x\in \mathcal{X}$, has self-crossings in $\mathcal{N}_i$, and  at least one of $G_x, G_y$ and $G_z$ is crossed exactly 3 times by $G_i$.}
\begin{proof}
Let $i\in \{1,\ldots, m\}$, and assume that $c_i$ contains the variables $x,y,z\in \mathcal{X}$ such that the ccw neighbors of $u_{5i}$ in $\gamma(H)$ are $(v_i,u^x_{5i-1}, u^y_{5i-1}, u^z_{5i-1}, u_{5i+1})$. Each of the graphs $G_x$, $G_y$, and $G_y$ have 3 vertex disjoint connected subgraphs in $\lambda^{-1}[H_i]$.
Due to the rotation of cluster $u_{5i+1}$ and $u_{5i+2}$, the path $G_i$ has to cross each of them, which yields at least 3 crossings in $\mathcal{N}_i$ with each graph $G_x, G_y$ and $G_z$. Furthermore, $G_x$, $G_y$, and $G_y$ each have 5 vertex disjoint connected subgraphs (each of which is formed by a single vertex) in $\lambda^{-1}[u_{5i+1}]$ (resp., $\lambda^{-1}[u_{5i+2}]$) with pipe-degree 2, and one with pipe-degree 1. For each $G_x$, $G_y$, and $G_y$ there exist altogether  exactly 7 edges incident to these subgraphs (vertices) in $\lambda^{-1}[u_{5i+1}u_{5i+2}]$. Note that $G_i$ has only one edge in $\lambda^{-1}[u_{5i+1}u_{5i+2}]$, which we denote by $e_i$.

 Without loss of generality we assume that  all the edge crossings of $G_i$ with $G_x,G_y$ and $G_z$ in the drawing $\Gamma\circ \Lambda$ occur along $e_i$, and outside of  $N_{u_{5i+1}u_{5i+2}}$.
By the latter, the drawing $\Gamma\circ \Lambda$ defines a total ``top to bottom'' order of the $7\cdot 3+1=22$ edges in $\lambda^{-1}[u_{5i+1}u_{5i+2}]$.
  Let $I_x$, $I_y$, and $I_z$ be the  minimum  intervals in this order spanned by the edges of $\lambda^{-1}[u_{5i+1}u_{5i+2}]$ in $G_x$, $G_y$, and $G_z$, respectively. If the edge $e_i$ is above (resp., below) all the 7 edges of $G_x$ in $\lambda^{-1}[u_{5i+1}u_{5i+2}]$, then
it creates at least 5 crossings with the edges incident to the pipe-degree 2 components in $N_{u_{5i+1}}$ (resp., $N_{u_{5i+2}}$). Analogous statements hold for $G_y$ and $G_z$, as well.
That is, if $e_i$ is not in $I_x$ (resp., $I_y$ and $I_z$), then $G_i$ crosses $G_x$ (resp., $G_y$ and $G_z$) at least 5 times in $\mathcal{N}_i$.

We distinguish several cases based on the relative positions of the intervals $I_x$, $I_y$, and $I_z$. If $I_x$, $I_y$, and $I_z$ are pairwise disjoint, then $e_i$ lies in at most one of these intervals, and $G_i$ crosses $G_x$, $G_y$, and $G_z$ altogether at least $3+5+5=13$ times.
If $e_i$ lies in exactly two of these intervals, say $I_x$ and $I_y$, then there are at least 2  crossings between $G_x$ and $G_y$ in $\mathcal{N}_i$,
and $G_i$ crosses $G_x$, $G_y$, and $G_z$ at least $3+3+5=11$ times. Finally, if $e_i$ lies in all three intervals, then there must be at least 6 crossings crossings between $G_x$, $G_y$, and $G_z$ in $\mathcal{N}_i$, and $G_i$ crosses $G_x$, $G_y$, and $G_z$ altogether at least $3+3+3=9$ times. In all cases, the number of crossings among $G_i$, $G_x$, $G_y$, and $G_z$ in $\mathcal{N}_i$ is at least 13, as required.
Equality is possibly only if none of $G_x$, $G_y$, and $G_z$ has self-crossings, and at least one of $G_x$, $G_y$, and $G_z$ is crossed by $G_i$ exactly 3 times.
\hfill$\Box$\end{proof}

\end{document}